\documentclass[onecolumn,journal]{IEEEtran}
\ifCLASSINFOpdf
\else
\fi

\usepackage[T1]{fontenc}
\usepackage[latin9]{inputenc}
\usepackage{color}
\usepackage{array}
\usepackage{float}
\usepackage{mathrsfs}
\usepackage{mathtools}
\usepackage{amsthm}
\usepackage{amstext}
\usepackage{amssymb}
\usepackage{stmaryrd}
\usepackage{graphicx}
\usepackage{pgfplots}

\makeatletter

\providecommand{\tabularnewline}{\\}

\theoremstyle{plain}
\newtheorem{thm}{\protect\theoremname}
\theoremstyle{plain}
\newtheorem{prop}[thm]{\protect\propositionname}
\theoremstyle{plain}
\newtheorem{cor}[thm]{\protect\corollaryname}
\theoremstyle{plain}
\newtheorem{lem}[thm]{\protect\lemmaname}
\theoremstyle{definition}
\newtheorem{example}[thm]{\protect\examplename}
\theoremstyle{definition}
\newtheorem{defn}[thm]{\protect\definitionname}

\AtBeginDocument{
	
}

\makeatother

  \providecommand{\corollaryname}{Corollary}
  \providecommand{\examplename}{Example}
  \providecommand{\lemmaname}{Lemma}
  \providecommand{\propositionname}{Proposition}
  \providecommand{\theoremname}{Theorem}
  \providecommand{\definitionname}{Definition}

\newcommand{\Div}{\operatorname{div}}
\newcommand{\Diff}{\operatorname{Diff}}

\newcommand{\res}{\operatorname{res}}
\newcommand{\supp}{\operatorname{supp}}
\begin{document}

%
\title{Multi-point Codes from Generalized Hermitian Curves}
%
%
%

\author{Chuangqiang~Hu and~Chang-An~Zhao    
	
\thanks{C. Hu and C.-A. Zhao are with the School
	of Mathematics and Computational Science, School of Mathematics and Computational Science, Sun Yat-sen University, Guangzhou 510275, P.R.China.\protect\\
	\protect\\
	E-mail: huchq@mail2.sysu.edu.cn,~{zhaochan3@mail.sysu.edu.cn}}
\thanks{Manuscript received *********; revised ********.}
} 

\maketitle

\begin{abstract}
We investigate multi-point algebraic geometric codes defined from curves related to the generalized Hermitian curve introduced by Alp Bassa, Peter Beelen, Arnaldo Garcia, and Henning Stichtenoth. Our main result is to find a basis of the Riemann-Roch space of a series of divisors, which can be used to construct multi-point codes explicitly. These codes
turn out to have nice properties similar to those of Hermitian codes, for example,  they are easy to describe, to encode and decode.  It is shown that the duals are also such codes and an explicit formula is given. In particular, this formula enables one to calculate the parameters of these codes. Finally, we apply our results to obtain linear codes attaining  new records on the parameters. A new record-giving
$ [234,141,\geqslant 59] $-code over $ \mathbb{F}_{27} $ is presented as one of the examples.
\end{abstract}


%
\IEEEpeerreviewmaketitle

\section{Introduction}
%
%
%
%
\IEEEPARstart{G}{oppa} constructed error-correcting linear codes by using tools from Algebraic Geometry: a non-singular, projective, geometrically irreducible, algebraic curve $\mathcal{X}$ of genus $g$ defined over $\mathbb{F}_{q}$, the finite field with $q$ elements, and two rational divisors $D$ and $G$ on $\mathcal{X}$. These divisors are chosen in such a way that they have disjoint supports and $D$ equals to a sum of pairwise distinct rational places, $D=P_{1}+\ldots+P_{n}$. The algebraic geometric code is defined as
\[C(\mathcal{X},D,G):={\left\{ (f(P_{1}), f(P_{2}),\ldots, f(P_{n})):f\in\mathcal{L}(G)\right\} },\]
where $\mathcal{L}(G)$ denotes the Riemann-Roch space associated to $ G $, see \cite{Stichtenoth} as general references for all facts concerning algebraic geometric codes.

One of the main features of Goppa's construction is that the minimum distance is
bounded from below, whereas in general there is no lower bound available on the minimum
distance of a code. The parameters of an algebraic geometric code are strictly  dependent on the curve chosen in the construction. If the curve possesses additional nice properties, one can hope that the corresponding algebraic geometric code also has nice properties. Algebraic geometric codes would advantage one to give an asymptotically good sequence of codes with parameters better than
the Varshamov-Gilbert bound in a certain range of the rate and for large enough
alphabets.

 The most studied codes are probably those arising from the Hermitian curve \cite{Stichtenoth,Tiersma}. The advantage
 of these codes is that the codes are easy to describe and to encode and decode. Moreover, these codes often have excellent parameters.

One-point codes on Hermitian curves were well-studied in
the literature, and efficient methods to decode them were known  \cite{Stichtenoth,Guruswami,Yang,Yang2}. The minimum distance of Hermitian two-point codes had been first
determined by M.~Homma and S.~J.~Kim  \cite{Homma,Homma2,Homma3,Homma4}. The explicit
formulas for the dual minimum distance of such codes were given by S.~Park  \cite{Park}. More recently,
Hermitian codes from higher-degree places had been considered in \cite{Korchmaros}. The dual minimum distance of many three-point codes from Hermitian curves was computed in \cite{Ballico}, by extending a recent and powerful approach by A.~Couvreur \cite{Couvreur}. S.~Bulygin investigated one-point codes from the generalized Hermitian curves proposed by Garcia and Stichtenoth; and  calculated some parameters of these codes \cite{Bulygin}. Some generalizations of these codes were studied by C.~Munuera, A.~Sep\'{u}lveda, and F.~Torres \cite{Munuera}.

In this paper we investigate multi-point codes from the other generalized Hermitian curves. Let $ q $ be a prime power, $ \mathbb{F}_{q^ {n_0}} $ be the finite field of order $ q^{n_0} $, with $ n_0 \geqslant 2 $, and $j_0,  k_0$ be two relatively prime positive numbers such that $ j_0+ k_0 =n_0 $. We are interested in algebraic geometric codes obtained from the non-singular model $\mathcal{X}$ over $ \mathbb{F}_{q^ {n_0}} $ of the plane curve
\begin{equation}
 \frac{y^{q^{j_0}}}{x}+\frac{y^{q^{j_0+1}}}{x^{q}}+\ldots+\frac{y^{n_0-1}}{x^{q^{k_0-1}}}+\frac{y^{1}}{x^{q^{k_0}}}+\frac{y^{q}}{x^{q^{k_0+1}}}+\ldots+\frac{y^{q^{j_0-1}}}{x^{q^{n_0}}}=1 \label{eq:the_plane_curve},
\end{equation}
which was introduced by Alp Bassa, Peter Beelen, Arnaldo Garcia, and Henning Stichtenoth \cite{Bassa}. Using Equation (\ref{eq:the_plane_curve}), they constructed towers of curves of large genera over $ \mathbb{F}_{q^{n_0}} $ which are optimal in the sense that they asymptotically attain the Drinfeld-Vladut bound. For $ n_0=2, j_0=k_0 =1 $ this is exactly the Hermitian curve over $ \mathbb{F}_{q^2} $. To see this,  we replace $ xy $ by $ z $ in Equation (\ref{eq:the_plane_curve}); then
\begin{equation}
z+z^q=x^{q+1},
\end{equation}
which is the usual definition of Hermitian curve \cite{Stichtenoth}. Our consideration is the special case with $ n_0 =3, j_0=1, k_0=2 $. According to the paper \cite{Bassa}, we introduce four divisors as follows:
\begin{enumerate}
		\item $D: = \sum_{\alpha, \beta } {D_{\alpha, \beta }} $, where $D_{\alpha,\beta}:=(x=\alpha, y=\beta)$ with $ \alpha , 
		\beta \in \mathbb{F}_ {q^3}^{*}$ satisfies Equation (\ref{eq:the_plane_curve});
		\item $V:=(x=0, y=\infty)$, the divisor consisting of all the places at infinity on the $y$-axis, which can be written as $V = \sum _ \mu V_{\mu}$, where $V_{\mu}:=(x=0,y=\infty,x^{q}y^{q+1}=\mu)$ represents a rational place in $ V $ and $\mu ^{q-1}=-1$ when $q$ is even;
		\item $Q:=(x=\infty, y= \infty )$, which contains a unique rational place just in case $ q $ is odd; 
		\item $P:=(x=0,y=0)$, the origin of the curve. 
\end{enumerate}
  The divisors $ D $, $ V $, $ Q $ and $ P $ contains all the rational places on the curve above. We define the algebraic geometric codes over   $ \mathbb{F}_{q^3} $ of even characteristic
	 \[
	C_{r}=C(\mathcal{X},D+P+V,rQ),\] 
	which are highly similar to the Hermitian codes. Also, we study the multi-point codes over   $ \mathbb{F}_{q^3} $ of arbitrary characteristic
	 \[
	 C_{r,s,t}=C(\mathcal{X},D,rQ+sP+tV).\]
	 To construct these codes, we find a basis for the Riemann-Roch space $ \mathscr L(rP+sQ+tV) $ by using Pick's theorem. We use Goppa bound and order bound to estimate the distances of these codes. It turns out that some such codes attain new record values on the parameters. By direct computation, a $[234,141,\geqslant 59]$-code over $ \mathbb F_{27} $ is presented as one of the examples.

	 The paper is organized as follows. In Section 2, we introduce some arithmetic properties of the curve (\ref{eq:the_plane_curve}). The properties of the codes $ C_r $ and $C_{r,s,t}$ are presented in Section 3 and Section 4 respectively. 

\section{the arithmetic properties of the curve}

We follow the notations in Section 1. Let $ q  $ be a power of a prime  and $ \mathbb{F}_{q^3} $ be a finite field of cardinality $ q^3  $. In this section we study the curve $ \mathcal{X} $ over $\mathbb{F}_{q^3} $
\begin{equation}
\frac{y^{q}}{x}+\frac{y^{q^{2}}}{x^{q}}+\frac{y}{x^{q^{2}}}=1\label{eq:curve_q=00003D3}.
\end{equation}
By the transformation $(x,y)\mapsto(1/y,1/x)$, we obtain 
\[
\frac{y^{q^{2}}}{x}+\frac{y}{x^{q}}+\frac{y^{q}}{x^{q^{2}}}=1,
\]
which is exactly Equation (\ref{eq:the_plane_curve}) with $j_0=2,k_0=1$.

Let $P:=(x=0,y=0)$, $Q:=(x=\infty,y=\infty)$, $V:=(x=0,y=\infty)$
be the divisors of the curve $ \mathcal{X} $ over $\mathbb{F}_{q^{3}}$.
Actually, $P$ is a rational place.
\begin{prop}
	\label{prop:valuation}
	\begin{enumerate}
		\item	The curve $ \mathcal{X} $ has genus $g=(q^{4}-3q+2)/2$.
		\item $\Div( x)=\Div_{0}(x)-\Div_{\infty}(x)=P+(q+1)V-qQ$, and $\Div( y)=\Div_{0}(y)-\Div_{\infty}(y)=q^{2}P-qV-Q$.
		\item $\deg(P)=1,\deg(Q)=q,\deg(V)=q-1$.
		\item $v_{P}(x^{-q^{2}}y)=0$, and $x^{-q^{2}}y \equiv 1 \mod P$. 
		\item If $ V_{\mu} $ is a rational place in V, then  $v_{V_{\mu}}(x^{q}y^{q+1})=0$, and $x^{q}y^{q+1}  \equiv \mu \mod V_{\mu}$,
			where $\mu^{q-1}= -1$.
		\item If $Q_{\delta}$ is a rational place in $ Q $, then  $v_{Q_{\delta}}(x^{-1}y^{q})=0$, and $x^{-1}y^{q}\equiv \delta \mod Q_{\delta}$,
			where $\delta+\delta^{q}=1$.
	\end{enumerate}
\end{prop}
\begin{proof}
	The assertions 1), 2), and 3) are shown in \cite{Bassa}. 
	
	4)The equation  $v_{P}(x^{-q^{2}}y)=0$ is clear by using assertion 2). It is easy to show that $v_{P}(\frac{y^{q}}{x})=q^{3}-1>0$, and $v_{P}(\frac{y^{q^{2}}}{x^{q}})=q^{4}-q>0$;
	then $1=\frac{y^{q}}{x}+\frac{y^{q^{2}}}{x^{q}}+\frac{y}{x^{q^{2}}}  \equiv \frac{y}{x^{q^{2}}} \mod P$.
	
	5)Multiplying both sides of Equation (\ref{eq:curve_q=00003D3}) by $x^{q}{}^{2}/y$, we get \[y^{q-1}x^{q^{2}-1}+(x^{q}y^{q+1})^{q-1}+1=\frac{x^{q^{2}}}{y}.\]
	Again by 2),  we have $v_{V_{\mu}}{(y^{q-1}x^{q^{2}-1})}=q^{3}-1>0$, and $v_{V_{\mu}}{(\frac{x^{q^{2}}}{y})}=q^{3}+q^{2}+q>0$.
	Hence, we find $(x^{q}y^{q+1})^{q-1}+1\equiv 0 \mod V_{\mu}$.
	
	6)Let $x^{-1}y^{q}\equiv \delta \mod Q_{\delta}$. Note that $v_{Q_{\delta}}(\frac{y}{x^{q^{2}}})=q^{3}-1>0$. This implies that $\delta+\delta^{q}=1$.
\end{proof}
We have described the rational places on the curve $\mathcal X$  in Section 1. The following corollary gives a simple explanation. 

\begin{cor}\label{cor:places}
	All the rational places on the curve are the following: $D_{\alpha,\beta}$,
		$P$, $Q_{\delta}$, if $q$ is odd; and $D_{\alpha,\beta}$,
		$P$,  $V_{\mu}$, if $q$ is even.
\end{cor}
\begin{proof}
By Proposition \ref{prop:valuation}, the divisor $ P $ is clearly a rational place. Note that the equation $\mu^{q-1}=-1$ has $ q-1 $ distinct roots in $ \mathbb{F}_{q^3} $ for even $q $ while it has no root for odd $q $. So  Proposition \ref{prop:valuation} implies that the divisor $ V $ can be written as $V = \sum _ \mu V_{\mu}$, where
$V_{\mu}:=(x=0,y=\infty,x^{q}y^{q+1}=\mu)$ represents a rational place in $ V $ when $q$ is odd. 

Not all the roots of the equation $\delta ^ q +\delta =1$ are contained in $ \mathbb{F}_{q^3} $. If $\delta+\delta^{q}=1$,
then $\delta^{q}+\delta^{q^{2}}=1$, and $\delta^{q^{2}-1}=1$. It gives $\delta\in\mathbb{F}_{q^{3}}^{*}\cap\mathbb{F}_{q^{2}}^{*}=\mathbb{F}_{q}^{*}$. Now the equation above becomes $2\delta=1$. Hence, it has exactly one root in $ \mathbb{F}_{q^3} $ if $ q $  is odd, and no root if $ q $ is even. Therefore, the unique rational place in $ Q $ can be written as $ Q_\delta : = (x=\infty, y=\infty, x^{-1}y^q=1/2 ) $. 

Observe that for all $ \alpha \in \mathbb {F}_{q^3}^{*} $, there exist exactly $ q^2 $ distinct elements $\beta \in \mathbb {F}_{q^3}^{*}$ with 
\begin{equation}
\frac{\beta^{q}}{\alpha}+\frac{\beta^{q^{2}}}{\alpha^{q}}+\frac{\beta}{\alpha^{q^{2}}}=1.\label{eq:alpha_beta}
\end{equation}
So there is a unique place $D_{ \alpha, \beta} $ of degree one such that
\[
x \equiv \alpha \mod D_{\alpha, \beta}, \quad y \equiv \beta \mod D_{\alpha, \beta}.
\]
Furthermore, the degree of the divisor $ D:=\sum D_{\alpha, \beta} $ is $ \deg D = (q^3-1)q^2 $.

\end{proof}
Now we come to determine the basis of the Riemann-Roch space $ \mathscr{L}(rQ+sP+tV) $.  A particularly favorable feature of Hermitian curves is that
one can explicitly write a monomial basis for the Riemann-Roch  space of a two-point divisor. This is the reason that Hermitian codes are easy  to encode and decode. Since the curve $ \mathcal{X} $ generalizes the Hermitian curves, we can expect to obtain a monomial basis of $ \mathscr{L}(rQ+sP+tV) $.  
  The following proposition is the main result of the paper which can be applied to encoding multi-point codes. 
\begin{prop}\label{prop:basis_rst}
	{The elements $ x^i y^j $ with $ (i,j)\in\Omega_{r,s,t} $  form a basis of $ \mathscr{L}(rQ+sP+tV) $, where 
		\begin{align*}
			\Omega_{r,s,t}:=\{ (i,j)|-t\leqslant(q+1)i-qj & <  q^{3}+q^{2}+q-t, \\
			-i-q^{2}j  &\leqslant s,\\
			qi+j &\leqslant r~\}.
		\end{align*} 
	}
\end{prop}
 We observe that determining the dimension of the Riemann-Roch space is equivalent to calculating the number of lattice points in some region. So our problem becomes a point-counting problem. The proof of the proposition will be given after some preparations.
\begin{lem}\label{lem:suppose}
Suppose $A=(x_{1},y_{1}),B=(x_{2},y_{2})$ are two lattice points on the plane
line $l_{r_0} :ax+by=r_{0}$, where $a$ and $b$ are two integers and coprime;
$l_{A},l_{B}$ are two lines pass through $ A $, $ B $ respectively and parallel
to each other; then for every plane line $l_{r} :ax+by=r$ parallel to $ l_{r_0}  $, the number $\# \sigma_r$ of the lattice
points within the segment $ \sigma_r$ between $ l_A $, $ l_B $ is a
constant $\frac{\left|x_{2}-x_{1}\right|}{\left|b\right|}$=$\frac{\left|y_{2}-y_{1}\right|}{\left|a\right|}$. (The reader shall be careful that we only count once if both two end-points of the segment are lattice points).
\end{lem}
\begin{proof}
	It is well known that the equation of $l_r$ has integer solutions if and only if $ r $ is  divisible by the greatest common factor of $ a $ and $ b $  \cite{Gilbert,Rosen}. Since  $ a $ and $ b $ are coprime, there are two integers $x'$ and $ y' $ such that $ ax'+by'=1 $. Hence, $ a r_0 x'+b r_0 y'=r_0 $. So we find a lattice point $ (x_0,y_0):=(r_0 x', r_0 y') $  on the line $l_r$.  
	
	We claim that  all the lattice points on the plane line $l_r$ are exactly $ \{(x_0+ t b, y_0-t a)| t \in \mathbb {Z} \}  $. Clearly, $ (x_0+ t b, y_0-t a) $ satisfies the equation of $ l_r $. Conversely, if $ (x,y) \in l_r $ is another lattice point, then 
	\[ a(x-x_0)+b(y-y_0)=0 .\]
	Since $ a $ and $ b $ are coprime, we get $ a|(y-y_0) $ and $ b|(x-x_0) $.
	
	Let $ C=(x_3,y_3) $ and $  D=(x_4,y_4) $ be two end-points of the segment $\sigma _ r$. Note that  the horizontal distance $ \left|  x_3-x_4 \right| $ between $ C $ and $ D $ is a constant. Moreover, the minimal horizontal distance between two distinct lattice points on $ \sigma_r $ is also a constant $b$.
	So the number $\# \sigma_r $ is  independent of $r$. Therefore, we only need to  find out all the lattice points on the segment $\sigma _{r_0}$. We see that all the lattice points on $ \sigma_{r_0} $ are 
	\begin{align*}
	&(x_1, y_1)=A, (x_1+ b, y_1 -a),(x_1+ 2 b, y_1 - 2 a),\\
	&\ldots, (x_1+ tb, y_1 -ta)=(x_2,y_2)=B ,
	\end{align*}
	 where $t=\frac{\left|x_{2}-x_{1}\right|}{\left|b\right|}$=$\frac{\left|y_{2}-y_{1}\right|}{\left|a\right|}$, which implies the lemma.
\end{proof}
\begin{figure}[H]
	\centering
	\begin{tikzpicture}[scale=0.65]
	\draw (0,0)--(4,4);
	\draw [dashed] (4,-2)--(8,2);
	 \draw [color=blue][<->] (0,-2)--(4,-2);%
	 \draw [color=blue][dotted] (0,0)--(0,-2);%
   	\draw [color=black] node at (1.5,-2.5) { $ \left| x_1-x_2 \right|  $ };
	\draw (0,0)--(4,-2);
	\draw [red] (3,3)--(7,1); 
    \draw [color=blue][<->] (3,1)--(7,1);%
    \draw [color=blue][dotted] (3,3)--(3,1);%
	\draw [color=black] node at (4.5,0.5) { $ \left| x_3-x_4 \right|  $ };
	\draw [fill] (0,0) circle [radius=0.1];
	\draw [fill] (1,-0.5) circle [radius=0.1];
	\draw [fill] (2,-1) circle [radius=0.1];
	\draw [fill] (3,-1.5) circle [radius=0.1];
		\draw [fill] (3.5,2.75) circle [radius=0.1];
		\draw [fill] (4.5,2.25) circle [radius=0.1];
		\draw [color=blue][dotted] (5.5,2.25)--(5.5,1.75);%
		\draw [color=blue][<->] (4.5,2.25)--(5.5,2.25);
		\draw[color=black] node at (4.95,2.55) { $ b $ };		
		\draw [fill] (5.5,1.75) circle [radius=0.1];
		\draw [fill] (6.5,1.25) circle [radius=0.1];
	\draw[color=black] node at (0,0) [left] {$ A $ }  ;
	\draw[color=black] node [yshift=-1ex,xshift=+1ex] at (4,-2){ $ B $ };
		\draw[color=black] node at (3,3) [left] {$ C $ }  ;
		\draw[color=black] node [yshift=-1ex,xshift=+1ex] at (7,1){ $ D $ };
	\draw[color=black] node [yshift=+1ex,xshift=-1ex] at (4,4) { $ l_A $ };
	\draw[color=black] node [yshift=+1ex,xshift=-1ex] at (8,2) { $ l_B $ };
	\draw[color=black] node [yshift=-1ex,xshift=+1ex] at (4,2) { $ \sigma_r $ };
	\draw[color=black] node [yshift=-1ex,xshift=+1ex] at (1,-1) { $ \sigma_{r_0} $ };
	\end{tikzpicture}
	\protect\caption{}
\end{figure}

\begin{prop}\label{prop: basis}
	\begin{enumerate}
		\item  The elements $ x^i y^j $ with $ (i,j)\in\Omega_{rQ}$ 
		form a basis of $ \mathscr{L}(rQ) $, where 
			\begin{align*}
					\Omega_{rQ}:=\{ (i,j)|0\leqslant(q+1)i-qj&<q^{3}+q^{2}+q, \\
				0 &\leqslant i+q^{2}j,\\
				qi+j &\leqslant r~\}.
			\end{align*}

		\item The elements $ x^i y^j $ with $ (i,j)\in\Omega_{sP}$  form a basis of $ \mathscr{L}(sP) $, where 
				\begin{align*}
					\Omega_ {sP}:=\{ (i,j)|0\leqslant(q+1)i-qj&<q^{3}+q^{2}+q, \\
					0&\leqslant-qi-j,\\
					-i-q^{2}j&\leqslant s~\}.
				\end{align*} 
		
		\item The elements $ x^i y^j $ with $ (i,j)\in\Omega_{tV}$  form a basis of $ \mathscr{L}(tV) $, where 
			\begin{align*}
				\Omega_{tV}:=\{ (i,j)|0\leqslant-qi-j&<q^{3}-1, \\
				0 &\leqslant i+q^{2}j,\\
				-(q+1)i+qj&\leqslant t~\}.
			\end{align*} 
		
	\end{enumerate}
	\begin{figure}[H]
		\centering
		\begin{tikzpicture}[scale=0.35]
		\path [fill=orange] (0,0)--(4,-1) -- (6,2) -- (4,6);
		\draw [blue] (-6,-9)--(6,9);
		\draw [dashed] (4,-1)--(8,5);
		\draw [blue](-8,2)--(8,-2);
		\draw [blue](-4.5,9)--(4.5,-9);
		\draw [blue] (2,3)--(4,-1); 
		\draw [red] (4,6)--(6,2); 
		\draw [dashed] (2,-4)--(-2,-10);
		\draw [red] (-4,-6)--(0,-7);		
		\draw [dashed] (-4,1)--(-7,7);
		\draw [red] (-6,5)--(-4,8);
		\draw[color=black] node at (0,0) [left] {$ O $ }  ;
		\draw[color=black] node [yshift=-1ex,xshift=+1ex] at (4,-1){ $ B $ };
		\draw[color=black] node [yshift=+1ex,xshift=-1ex] at (2,3) { $ A $ };
		\draw[color=black] node [yshift=-1ex,xshift=+1ex] at (6,2){ $ D $ };
		\draw[color=black] node [yshift=+1ex,xshift=-1ex] at (4,6) { $ C $ };
		\draw[color=black] node  at (8,3.5) { $ l_B $ };
		\draw[color=black] node  at (5.05,8.57) { $ l_A $ };
		\draw[color=black] node [yshift=+1ex,xshift=-3ex] at (-3,-4) { $ (q+1)i-qj=0 $ };
		\draw[color=black] node [yshift=+1ex,xshift=-2ex] at (-4,0) { $ i+q^2 j=0 $ };
		\draw[color=black] node [yshift=+1ex,xshift=-2ex] at (6,-8) { $  qi+j=0 $ };
		\draw[color=black] node [yshift=-1ex,xshift=-1ex] at (6,5) { $ l_r $ };
		\draw[color=black] node [yshift=-1ex,xshift=-1ex] at (0,-5) { $ \Omega_{sP} $ };
		\draw[color=black] node [yshift=-1ex,xshift=-1ex] at (-3,5) { $ \Omega_{tV} $ };
	\draw[color=black] node [yshift=-1ex,xshift=-1ex] at (4.5,3) { $ \Omega_{rQ} $ };
		\draw [fill] (0,0) circle [radius=0.1];
		\draw [fill] (2,3) circle [radius=0.1];
		\draw [fill] (4,-1) circle [radius=0.1];
		\end{tikzpicture}
		\protect\caption{}
	\end{figure}
	
\end{prop}
\begin{proof}
	We prove only the first assertion of this proposition. The other conclusions can be deduced similarly.
	Proposition \ref{prop:valuation} implies
	\begin{align*}
		\Div( x^i y^j ) & =  i\Div( x)+j \Div( y)\\
		& =  iP+(q+1)iV-qiQ+q^{2}j P-q j V-j Q\\
		& =  (i+q^{2}j)P+\left( (q+1) i- q j \right) V-(q i+ j)Q.
	\end{align*}
	Thus, $ x^i y^j \in \mathscr{L}(rQ) $ if and only if $ 0\leqslant(q+1)i-qj$, $ 0\leqslant i+q^{2}j$, and $qi+j\leqslant r $. Hence, all the elements in  $\left\{ x^{i}y^{j}|(i,j)\in\Omega_{rQ}\right\} $ are contained in $ \mathscr{L}(rQ) $. 
	
 Similar to the proof of the above lemma, we assume $(i,j)\in\Omega_{rQ}$, then the valuation of  $x^{i}y^{j}$ at the place $P$ is $ i+q^{2}j $. The
	element $x^{k}y^{l}$ with the same valuation at $P$ satisfies 
	\[
	k=\lambda q^{2}+i,~l=-\lambda+j.
	\]
	By definition,
	\[
	0\leqslant(q+1)i-qj<q^{3}+q^{2}+q,
	\]	
	and 
	\begin{equation*}
			 (q+1) k-q l  =  (q+1) i- q j+ \lambda \left(q^{3}+q^{2}+q\right).
	\end{equation*}
	Hence, $(k,l)$ is outside the set $\Omega_{rQ}$ for $\lambda \not = 0$. It follows that
	all the elements in $\left\{ x^{i}y^{j}|(i,j)\in\Omega_{rQ}\right\} $
	have different valuations at the place $P$, therefore they are linearly independent. To complete the proof, we only need to show that the number $ \#\Omega_{rQ} $ of the set $ \Omega_{rQ} $ is exactly the dimension of $  \mathscr{L}(rQ)  $ for every $ r\geqslant 0 $.
	
	 Let $A=(q^{2}-q,q^{2}-1)$, $B=(q^{2},-1)$, and $O=(0,0)$. Denote by $l_A$ and $l_B$  the  parallel lines   $  (q+1)i-qj=0 $ and $  (q+1)i-qj=q^3+q^2+q $ respectively; and denote by $C$ and $D$ the intersection  points of  the line $l_r: qi+j =r$ and the parallel lines  $ l_A $, $ l_B $. By definition, the set $ \Omega_{r_{0} Q} $ contains exactly the lattice points in the trapezoid $ OBDC $ except the edge $ BD $.  Applying  Pick's Theorem \cite{Haigh,Varberg}, the number $I $ of the lattice points
	in the interior located in the triangle $\triangle OAB$ can be calculated
	by the formula 
	\[
	I=S-M/2+1,
	\]
	where $ S $ is the area of the triangle $\triangle OAB$,
	and $ M $ is the number of lattice points on the boundary. Let $r_{0}:=q^{3}-1$.  The equation of the line $ AB $ is $ l_{r_0}:qi+j=r_0 $. Note that the lattice points in $ \Omega_{r_{0} Q} $ are those in the triangle $\triangle OAB$ except the vertex $ B $. Hence, we have
	\[ \#\Omega_{r_{0} Q}=I + M -1= S +M /2. \]
    It follows from the proof of Lemma \ref{lem:suppose} that  all the
	lattice points on the segment $\overline{OA}$ are 
	 \begin{align*}
	&(0,0)=O, \left(q,q+1\right), \left(2q,2(q+1)\right), \left(3q,3(q+1)\right),\\
	&\ldots,\left( (q-2) q,(q-2)(q+1) \right), \left(q^{2}-q,q^{2}-1\right)=A; 
	\end{align*} 
	and those on the segment $\overline{AB}$ are
	 \begin{align*}
	 &\left(q^{2}-q+1,(q^{2}-1)-q\right), \left(q^{2}-q+2,(q^{2}-1)-2q\right),\\
	 &\ldots,\left(q^{2}-1,q-1\right),\left(q^{2},-1\right)=B.
	 \end{align*} 
	Moreover, there is no lattice point on the segment $\overline{OB} $ except the end-points. Hence, $M =2q$. By direct computation, we find $ S =(q^{2}-1/q)q^{2}/2=q^{4}/2-q/2$. This implies 
	\[ \#\Omega_{r_{0} Q}=q^{4}/2+q/2 .\] On the other hand,	
	since $\deg(r_{0}Q)=q^{4}-q>q^{4}-3q=2g-2$, by the Riemann-Roch Theorem, we obtain
	\begin{align}
	\dim\mathscr{L}(r_{0}Q)&=1-g+\deg(r_{0}Q)\nonumber \\
	&=-(q^{4}-3q)/2+(q^{3}-1)q\nonumber\\
	&=q^{4}/2+q/2=\#\Omega_{r_{0}Q}\label{eq:dim_eq_O},
	\end{align}
	and
	\[
	 \dim\mathscr{L}((r+1)Q)=\dim\mathscr{L}(rQ)+q,\quad \text {for 
	 	$r \geqslant r_0$} .
	\]
	Note that the set consisting of the lattice points on the segment $\overline{CD} $ is equal to $ \Omega_{(r+1)Q} \setminus \Omega_{rQ} $. Now Lemma \ref{lem:suppose} shows that $\#\Omega_{(r+1)Q}=\#\Omega_{rQ}+q$
	for $r\geqslant r_{0}$, therefore $\#\Omega_{rQ}=\dim\mathscr{L}(rQ)$.
	
	It remains to consider the case $r<r_0 $.
	Since the lattice point $ (i,j)\in \Omega_{rQ}\setminus \Omega_{(r-1)Q}  $ represents an element in $ \mathscr{L}(rQ) \setminus \mathscr{L}((r-1)Q $, we obtain 		
    \begin{equation}
	\#\Omega_{rQ} -\#\Omega_{(r-1)Q}\leqslant \dim\mathscr{L}(rQ)-\dim\mathscr{L}((r-1)Q).\label{eq:Omega_eq_dim}
	\end{equation}
	  Sum both sides of Equation (\ref{eq:Omega_eq_dim}). Then $ \#\Omega_{r_0 Q} \leqslant \dim\mathscr{L}(r_0 Q) $. Moreover, Equation (\ref{eq:dim_eq_O}) implies that the equality holds in (\ref{eq:Omega_eq_dim}) for $r\leqslant r_0$. So we conclude that $\#\Omega_{rQ}=\dim\mathscr{L}(rQ)$ for $r\leqslant r_0$.
	\end{proof}
\begin{cor}\label{cor:basis_rQ}
	The elements $ x^i y^j $ with $ (i,j)\in\Omega_{rQ}^{\prime} $ 
	form a basis of $ \mathscr{L}(rQ) $, where 
		\begin{align*}
			\Omega_{rQ}^{\prime}:=\{ (i,j)|0 &\leqslant (q+1)i-qj,~qi+j\leqslant r,   \\
			-1 &<j\leqslant q^{2}-1 ~\} .
		\end{align*}
	\end{cor}
	\begin{proof}
		As the statement in the proof  of Proposition \ref{prop: basis}, all the elements in $\left\{ x^{i}y^{j}|(i,j)\in\Omega^\prime_{rQ}\right\} $ are contained in $ \mathscr{L}(rQ) $.
		According to \cite{Bassa}, the polynomial 
		\begin{equation*}
		\phi(y):=\frac{y^{q^{2}}}{x^{q}}+\frac{y^{q}}{x}+\frac{y}{x^{q^{2}}}-1,
		\end{equation*} over $\mathbb{F}_{q^{3}}(x)$ is irreducible. 
		Then $1,y^{1},y^{2},\ldots, y^{q^{2}-1}$ are linearly independent over $\mathbb{F}_{q^{3}}(x)$. So all the elements contained in $\left\{ x^{i}y^{j}|(i,j)\in\Omega^\prime_{rQ}\right\} $ are linearly independent over $\mathbb{F}_{q^{3}}$.  Denote the number of the lattice point set $ \Omega_{rQ}^{\prime} $  by $ \# \Omega_{rQ}^{\prime} $. It is sufficient to show that 
		  $ \#\Omega_{rQ}^{}=\#\Omega_{rQ}^{\prime} $. From the figure below, we see that $ \Omega_{r Q}^{\prime} $ (resp. $ \Omega_{r Q} $) contains  exactly the lattice points in the polygon $OBDC$ (resp. $ OBD'C'A $) except  the edge $BD$ (resp. $ BD' $); and in particular $ \Omega_{r_0 Q}^{\prime} $ (resp. $ \Omega_{r_0 Q} $) contains  exactly the lattice points in the triangle $\triangle OAB$ except the vertex $B$. Thus,  $ \#\Omega_{r_0 Q}^{}=\#\Omega_{r_0 Q}^{\prime} $. Now the corollary follows from Lemma \ref{lem:suppose}.
	\end{proof}
	\begin{figure}[H]
		\centering
		\begin{tikzpicture}[scale=0.35]
		\path [fill=orange] (0,0)--(4,-1)--(8,5)--(6,9);
		\path [fill=yellow] (0,0)--(4,-1)--(11,-1)--(9,3)--(2,3);
		\draw  [blue] (-2,-3)--(7,10.5);
		\draw (2,3)--(10,3);
		\draw [dashed] (4,-1)--(9,6.5);
		\draw [dashed] (4,-1)--(12,-1);
		\draw [blue](-8,2)--(10,-2.5);
		\draw [blue](-4.5,9)--(2,-4);
		\draw [blue] (2,3)--(4,-1); 
		\draw [red] (6,9)--(8,5); 
		\draw [red] (9,3)--(11,-1); 
		\draw[color=black] node at (0,0) [left] {$ O $ }  ;
		\draw[color=black] node [yshift=-1ex,xshift=+1ex] at (4,-1){ $ B $ };
		\draw[color=black] node [yshift=+1ex,xshift=-1ex] at (2,3) { $ A $ };
		\draw[color=black] node [yshift=-1ex,xshift=+1ex] at (8,5){ $ D $ };
		\draw[color=black] node [yshift=+1ex,xshift=-1ex] at (6,9) { $ C $ };
		\draw[color=black] node [yshift=-1ex,xshift=+1ex] at (11,-1){ $ D' $ };
		\draw[color=black] node [yshift=+1ex,xshift=-1ex] at (9,3) { $ C' $ };
		\draw[color=black] node [yshift=+1ex,xshift=-2ex] at (-3,-4) { $ (q+1)i-qj=0 $ };
		\draw[color=black] node [yshift=+1ex,xshift=-2ex] at (-4,0) { $ i+q^2 j=0  $ };
		\draw[color=black] node [yshift=+1ex,xshift=-2ex] at (-4,6) { $ qi+j=0  $ };
		\draw[color=black] node [yshift=-1ex,xshift=-1ex] at (6,5) { $ \Omega_{rQ} $ };

		\draw[color=black] node [yshift=-1ex,xshift=-1ex] at (9,1) { $ \Omega_{rQ}^\prime $ };
				\draw [fill] (0,0) circle [radius=0.1];
				\draw [fill] (2,3) circle [radius=0.1];
				\draw [fill] (4,-1) circle [radius=0.1];
		\end{tikzpicture}
		\protect\caption{}
	\end{figure}
	Now we are in a position to give the proof of Proposition \ref{prop:basis_rst}.
		\begin{proof}[Proof of Proposition \ref{prop:basis_rst}]
			As shown in the figure below, we have
			\[  \Omega_{r,s,t}=D_{1}'\cup D_{2}\cup D_{3}\cup D_{4}'\cup D_{5}'\cup D_{6}\cup D_{7} , \] 
			and 
			\[ \Omega_{rQ}=D_{1}\cup D_{2}\cup D_{3}\cup D_{4}. \]
			Clearly, $ \#D_{i}=\#D_{i}' $, so we obtain 
			\begin{align*}
			\#\Omega_{r,s,t}&=\#\Omega_{rQ}+\#(D_{5}\cup D_{6})+\#D_{7}\\
			&=\#\Omega_{rQ}+s+(q-1)t\\
			&=1-g+\deg(rQ+sP+tV).	
			\end{align*}
			The Riemann-Roch Theorem implies 
			\[ \#\Omega_{r,s,t}	= \dim \mathscr{L}(rQ+sP+tV).\]
			As the statement in the proof of Proposition \ref{prop: basis}, we find that  $\left\{ x^{i}y^{j}|(i,j)\in\Omega_{r,s,t}\right\} $ is
			a basis of $\mathscr{L}(rQ+sP+tV)$.
			\end{proof}
				\begin{figure}[H]
					\centering
					\begin{tikzpicture}[scale=0.35]
					\path [fill=orange] (0,0)--(4,-1)--(6,2)--(4,6);
					\path [fill=yellow] (1.4,-1.75)--(1.4-4,-1.75+1)--(4-1*0.91,6+2*0.91)--(6-1*0.91,2+2*0.91);
					\draw  [blue] (-2,-3)--(6,9);

					\draw [dashed] ( 3.2,  -2.2 )--(9,6.5);
					    \draw [blue] ( -0.8, -1.2 )--(3.2, -2.2 );
						\draw [black, dashed] (1.4,-1.75)--(1.4+7,-1.75+7/2*3);
						\draw [black] (1.4,-1.75)--(1.4-4,-1.75+1);
						\draw [black] (1.4-4,-1.75+1)--(1.4-4+7,-1.75+1+7/2*3);
					\draw [blue](-8,2)--(8,-2);			
					\draw [thin, dashed, blue](-4.5,9)--(2,-4);
					\draw [blue] (2-1*0.91,3+2*0.91)--(4,-1); 
					\draw [red] (4-1*0.91,6+2*0.91)--(6,2);

					\draw[color=black] node [yshift=-1ex,xshift=-3ex] at (4,-1){ $ D_5 $ };
					\draw[color=black] node [yshift=1.35ex,xshift=-2ex] at (0,0){ $ D_1' $ };
					\draw[color=black] node [yshift=-1.25ex,xshift=-3ex] at (0,0){ $ D_5' $ };
						\draw[color=black] node [yshift=-2.25ex,xshift=2ex] at (0,0){ $ D_6 $ };

					\draw[color=black] node  at (3.1,-0.4) { $ D_1 $ };
					\draw[color=black] node  at (0.5,2.6) { $ D_7 $ };
					\draw[color=black] node  at (1.7,0.9) { $ D_2 $ };
					\draw[color=black] node  at (3.1+1.8,-0.4+2.5) { $ D_4 $ };
					\draw[color=black] node  at (0.5+2,2.6+3) { $ D_4' $ };
					\draw[color=black] node  at (1.7+2,0.9+3) { $ D_3 $ };
						\draw[color=black] node [yshift=+1ex,xshift=-2ex] at (-3,-4) { $ (q+1)i-qj=0 $ };
						\draw[color=black] node [yshift=+1ex,xshift=-2ex] at (-4,0) { $ i+q^2 j=0 $ };
						\draw[color=black] node [yshift=+1ex,xshift=-2ex] at (-4,6) { $ qi+j=0 $ };
							\draw [fill] (0,0) circle [radius=0.1];
							\draw [fill] (2,3) circle [radius=0.1];
							\draw [fill] (4,-1) circle [radius=0.1];
					\end{tikzpicture}
					\protect\caption{}
				\end{figure}
	We remake that our method in the proof above can be applied to determining the basis of $\mathscr{L}(rQ+sP+tV) $ on the curve (\ref{eq:the_plane_curve}) with $ j_0=1,k_0=n_0-1 $.
\section{the code $ C(\mathcal{X},D+P+V,rQ) $}
Let $ E:=D+P+V $. Throughout this section, we assume that $ q $ is even. Then the divisor $V$ consists  of rational places by Corollary \ref{cor:places}. We study the linear code
 \[
 C_{r}=C(\mathcal{X},E,rQ).
 \] 
 The length of $C_{r}$ is $n_{1}:=\deg(E)=\deg(D+P+V)=1+q-1+(q^{3}-1)q^{2}=q(q^{4}-q+1)$.
 It is well known that the dimension of an algebraic geometric code $ C(\mathcal{X},E,G) $ is given by
 \begin{equation}
 \dim C_{r}=\dim\mathscr{L}(G)-\dim \mathscr{L}(G-E).\label{eq:dim}
 \end{equation}
Let $ R_1: = (n_{1}+2g-2)/q=q^{4}+q^{3}-q-2 $. For $r>R_1$,  the Riemann-Roch Theorem and Equation (\ref{eq:dim}) yield 
  \begin{align*}
  \dim C_{r}&=(1-g+\dim(rQ))-(1-g+\deg(rQ-E))\\
  &= \deg E = n_1.
  \end{align*}
  Hence $ C_{r} = \mathbb {F}_{q^3}^{n_1}$ in this case which is trivial. So we should only consider the case $0 \leqslant r \leqslant R_1$. 
  
\begin{defn}
	 Two codes $ C_1, C_2 \subseteq \mathbb {F}^{n}_{q^3} $ are said to be \textbf{equivalent}  if there is a vector $ a =(a_1, a_2, \ldots, a_n ) \in (\mathbb {F}_{q^3}^{*})^n $ such that $ C_2 = a \cdot C_1 $; i.e.,
	 \[
	 C_2 = \left\{(a_1 c_1, a_2 c_2, \ldots , a_n c_n)|(c_1, c_2,\ldots, c_n)\in C_1\right\}.
	 \] 
	 Denote by $ C^\bot $ the dual of $ C $. The code $ C $ is called \textbf{self-dual} (resp. \textbf{self-orthogonal}) if $ C =  C ^\bot$ (resp. $ C \subseteq  C ^\bot$). The code $ C $ is called \textbf{self-equivalent} if $ C $ is equivalent to  $ C ^\bot$.
\end{defn}

  We need the following lemma which is shown in \cite{Stichtenoth}.
\begin{lem}[\cite{Stichtenoth}]\label{lem:dual}
	Let $\eta$ be a Weil differential such that $v_{P_{i}}(\eta)=-1$
		for $i=1,\ldots,n$. Then the dual of $  C({\mathcal{X}},D,G) $ is equivalent to the code $ C({\mathcal{X}},D,D-G+\Div( \eta)) $. Furthermore, denote by $ \res_{P}(\eta)$ the residue of $\eta $ at $P $, then each vector in $ C({\mathcal{X}},D,G)^{\bot}  $ can be written as 
		\[
		\left( \res_{P_{1}}(\eta)c_{1},\res_{P_{2}}(\eta)c_{2},\ldots,\res_{P_{n}}(\eta)c_{n} \right) ,
		\]
		where $\left( c_{1},c_{2},\ldots,c_{n} \right) \in  C({\mathcal{X}},D,G) $.
		Moreover, if {$\res_{P_{i}}(\eta)=1$,
		then the dual of $C({\mathcal{X}},D,G)$ is equal to $C({\mathcal{X}},D,D-G+\Div( \eta))$.}\end{lem}
\begin{prop}
	\label{prop:The-dual-of}The dual of $C_{r}$ is
		\[ C_{r}^{\bot}=C_{R_1-r}. \]
	 Hence $C_{r}$ is self-orthogonal if $ 2r \leqslant R_1 $, and $C_{r}$ is self-dual for $ r = R_1/2 $.
\end{prop}
\begin{proof}
	Proposition \ref{prop:valuation} shows
	\begin{align*}
		\Div( x) & =  \Div_{0}(x)-\Div_{\infty}(x)\\
		& =  P+(q+1)V-q Q,
	\end{align*}
	and
	\begin{align*}
		\Div( y) & =  \Div_{0}(y)-\Div_{\infty}(y)\\
		& =  q^{2} P-q V- Q.
	\end{align*}
	Consider the element
	\[
	t:=\prod_{\alpha\in\mathbb{F}_{q^3}}(x-\alpha)=x^{q^{3}}-x.
	\]
    Then $t$ is a prime element for all places $ D_{\alpha, \beta} $, and its  divisor is 
	\begin{align*}
		\Div( t) & =  \Div_{0}(x)+D-q^{3}\Div_{\infty}(x)\\
		& =  P+(q+1)V+D-q^{4}Q.
	\end{align*}
	The differential $ dt $ has the divisor 
	\begin{align*}
		\Div( dt) & =  \Div( -dx)=-2\Div_{\infty}(x)+\Diff(F/K(x))\\
		& =  -2 q Q +(q^{3}+q-2)Q+qV\\
		& =  (q^{3}-q-2)Q+qV,
	\end{align*}
	where we use the formula  	$ \Diff(F/K(x))=(q^{3}+q-2)Q+qV $ according to \cite{Bassa}. Let $ \eta := dt/t $ be a Weil differential. The divisor of $ \eta $ is
	\begin{align*}
	\Div (\eta) &= \Div(dt)-\Div(t)\\
		& =  (q^{3}-q-2)Q+qV
		\\
		& \quad  -P-(q+1)V-D+q^{4}Q\\
		& = -P-V-D+(q^{4}+q^{3}-q-2)Q\\
		& = -P-V-D+R_1 Q.
	\end{align*}
	Clearly, the Weil differential $ \eta  $ satisfies the condition in  Lemma \ref{lem:dual}; therefore the dual $ C_{r}^{\bot}  $ of $ C_{r} $ is equivalent to 
	\begin{align*}
	   C\left({\mathcal{X}},E,E -r Q +\Div(\eta)\right)
		& =  C\left({\mathcal{X}},E,P+D+V-r Q\right.\\
		& \quad  \left.-P-V-D+R_1 Q\right)\\
		& =  C\left({\mathcal{X}},E,(R_1 -r)Q\right)\\
		& =  C_{R_1 -r} .
	\end{align*}
	It remains to show that $ \res_{P'}(\eta)=1  $, for cases $ P '=P$, $P'=D_{\alpha,\beta} $, and $P'= V_\mu$. Only the last case $ P '= V_\mu $ is non-trivial. Define  	$ z:=xy $. Then $z$ is prime at $ V_\mu $ and its divisor is
	\begin{align*}
		\Div( z) & =  \Div( x)+\Div( y)\\
		& =  P+(q+1) V - qQ+q^{2}P- q V- Q \\
		& =  (1+q^{2})P+V-(q+1)Q.
	\end{align*}
	By Equation (\ref{eq:curve_q=00003D3}), 
	\[ x^{q+q^{2}}y^{q}+x^{1+q^{2}}y^{q^{2}}+x^{1+q}y=x^{1+q+q^{2}}. 
	\]
	Replacing $ xy $ by $ z  $, we get
	\begin{equation}
	x^{q^2}z^{q}+xz^{q^{2}}+x^{q}z=x^{1+q+q^{2}}\label{eq:curve_x_z}.
	\end{equation}
	The differential of Equation (\ref{eq:curve_x_z}) is	
	\[
	-\left(z^{q^{2}}-(1+q+q^{2})x^{q+q^{2}}\right)dx=x^{q}dz.
	\]
	This implies 
	\begin{align*}
	 \eta &=	\frac{dt}{t}  =  \frac{-dx}{x^{q^3}-x}\\
		&=\frac{x^{q}dz}{\left(z^{q^{2}}-(1+q+q^{2})x^{q+q^{2}}\right)(x^{q^3}-x)}\\
		& =  \frac{x^{q-1}dz}{\left(z^{q^{2}}+*\right)(-1+*)}\\
		& =  \left(-\frac{x^{q-1}}{z^{q^{2}-1}}+*\right)\frac{dz}{z}.
	\end{align*}
	By Proposition \ref{prop:valuation}, we obtain $\frac{x^{q-1}}{z^{q^{2}-1}}=(x^{q}y^{q+1})^{q-1}\equiv-1 \mod V_\mu $.
	Hence, $\res_{V_{\mu}}(\eta)=\res_{V_{\mu}}({dt}/{t})=1$. Now,  Lemma \ref{lem:dual}
	shows $C_{r}^{\bot}=C_{R_1-r}$.\end{proof}
\begin{prop}\label{prop:dimension}
	{{Suppose that}} {$0\leqslant r\leqslant R_{1}$. Then the following holds:}\end{prop}
\begin{enumerate}
	\item The dimension of $C_{r}$ is given by 
	\[
	\dim C_{r}=\begin{cases}
	\#\Omega_{rQ}^{\prime} \quad \text{ for \ensuremath{0\leqslant r<q^{4}-q+1}},\\
	n_{1}-\#\Omega_{sQ}^{\prime}=q(q^{4}-q+1)-\#\Omega_{sQ}^{\prime}  \\
	 \qquad \qquad \text{for }q^{4}-q+1\leqslant r\leqslant R_1.
	\end{cases}
	\]	
	where $s=R_1-r$.
	
	\item For $q^{3}-3<r<q^{4}-q+1$ we have $\dim C_{r}=qr-(q^{4}-3q)/2$.

	\item The minimum distance $d$ of $C_{r}$ satisfies $d\geqslant q(q^{4}-q+1-r)$.\end{enumerate}
\begin{proof}
	
	\begin{enumerate}
		\item For $ 0\leqslant r<n_{1}/q=q^{4}-q+1 $, the inequality $\deg (rQ-E)<0$ implies $\mathscr{L}(E-rQ)=0$. By  Proposition \ref{cor:basis_rQ} and Equation (\ref{eq:dim}), we get
		 \[ \dim C_{r}=\dim\mathscr{L}(rQ)=\#\Omega_{rQ}^{\prime} .\]
		For $\ensuremath{q^{4}-q+1}\leqslant r\leqslant q^{4}+q^{3}-q-2$,
		we set $s:=q^{4}+q^{3}-q-2-r$, then
		\[
		0\leqslant s=q^{4}+q^{3}-q-2-r\leqslant q^{3}-3<q^{4}-q+1.
		\]
		Proposition \ref{prop:The-dual-of} yields 
		\[ \dim C_{r}=n_{1}-\dim C_{s}=q(q^{4}-q+1)-\#\Omega_{rQ}^{\prime}.  \]
		\item For $q^{3}-3<r<q^{4}-q+1$, $ \deg(rQ)>2g-2 $, so the Riemann-Roch Theorem gives
		\[
		\dim C_{r}=qr+1-g=qr-(q^{4}-3q)/2.
		\]
		\item The inequality $d\geqslant q(q^{4}-q+1-r)$ follows from Goppa bound.
	\end{enumerate}
\end{proof}
\begin{prop}
	{\label{prop:Suppose_Bound_d}Suppose $0\leqslant r<q^{4}-q+1$}.
\begin{enumerate}
	\item If $r=qt$, $t\leqslant q^{3}-1$, then $d=q(q^{4}-q+1-r)$.
	\item If $r=1+q$, then $d=q(q^{4}-q+1-r)$. 
	\item If $r=1+qt$, $0<t\leqslant q^{3}-q^{2}$, then $d=q(q^{4}-q+1-r)$. 
	\item If $r=(1+q)+(q^{3}+q^{2}+q)t$,  	$0 \leqslant  t \leqslant q-2 $, then $d=q(q^{4}-q+1-r)$. 
 \end{enumerate}
\end{prop}
\begin{proof}
	
	\begin{enumerate}
		\item Choose $t$ distinct elements $\alpha_{1},\ldots,\alpha_{t}\in\mathbb{F}_{q^{3}}^{*}$,
		and consider the element 
		\[
		z:=\prod_{i=1}^{t}(x-\alpha_{i}).
		\]
		Its divisor is
		\[ \Div( z)=-qtQ+\sum_{i=1}^{t}D_{\alpha_{i}}, \]
		where $ D_{\alpha_{i}} $ denotes the divisor ${D_{\alpha_{i}}}:=\sum_{\beta}D_{\alpha_i,\beta}$.
	 Therefore, $z\in\mathscr{L}(rQ)$.
		Note that the element $z$ has exactly $q^{2}t$ distinct zeros $D_{\alpha_i,\beta}$
		of degree one, so the weight of the corresponding codeword $ev_{D}(z)\in C_{r}$
		is $q(q^{4}-q+1-qt)$.
		\item Fix an element $\beta \in   \mathbb{F}_{q^3}^{*} $, and consider 
		\[
		z:=x(y-\beta).
		\]
		By the strict triangle inequality, we have
		\[ \Div( y-\beta)=-Q+D_{\beta}-qV, \]
		 where $ D_{\beta} $ denotes the divisor $D_{\beta}:=\sum_{\alpha}D_{\alpha,\beta}$ with $\alpha, \beta $ satisfying Equation (\ref{eq:alpha_beta}). The  degree of $ D_{\beta} $ is
		$ \deg(D_{\beta})=q^{2} $. We find the divisor of $ z $ is
		\[  \Div( z)=-(q+1)Q+D_{\beta}+P+V. \]
		\item Let $ \beta \in \mathbb {F} ^ {*}_{q^3} $, and $A$ be the set $\{\alpha|(\alpha,\beta)\not \in D_{\beta} \text { be the solution of Equation (\ref{eq:alpha_beta}) }\}$. Since the cardinality of $ A $ is $ \# A = q^3-1-q^2 $, we can choose $ s=t-1 $ distinct elements $\alpha_i \in A$. Consider 
		\[
		z:=x(y-\beta)\prod_{i=1}^{s}(x-\alpha_{i}).
		\]
		The divisor of $z$  is
		\[ \Div( z)=-(q t +1)Q+D_{\beta}+P+V+\sum_{i=1}^{s}D_{\alpha_{i}}. \]
		
		\item If $0 \leqslant t \leqslant q-2$, then $1+(q^{2}+q+1)t\leqslant q^{3}-1$. We can choose $ s=1+(q^2+q+1)t $  distinct elements $\beta_i \in \mathbb {F}_{q^3}^{*}$ and consider 
		\[
		z:=x^{1+q^{2}t}y^{-t}\prod_{i=1}^{s}(y-\beta_{i}).
		\]
	    Its divisor is 
	    \[\Div( z)=-(q+1+(q^{3}+q^{2}+q)t)Q+\sum_{i=1}^{s} D_{\beta_{i}}+P+V  .\]
	\end{enumerate}
\end{proof}
\begin{prop}
	{\label{prop:Suppose-that-,}Suppose that $q^{4}+q^{3}-q^{2}-2q-3<r\leqslant R_1$, then the minimum distance of $ C_r $ is 		$d=2$.}\end{prop}
\begin{proof}
	The element in $\Omega_{rQ}$ with least valuation at $Q$, and valuation
	$0$ at $V$ is $x^{q}y^{q+1}$. Let $r_{1}:=-v_{Q}(x^{q}y^{q+1})=q^{2}+q+1$.
	Assume that $q^{4}+q^{3}-q^{2}-2q-3<r\leqslant R_1 $. 
	By Proposition \ref{prop:The-dual-of}, the dual of $C_{r}$ is $C_{R_1-r}$ with $0\leqslant R_1-r<q^{2}+q+1=r_{1}$.
	According to Corollary \ref{cor:basis_rQ}, the basis of $C_{R_1-r}$ can be given by $f_1=1,f_2=x,f_3=xy,\ldots ,f_k $, where $ k $ denotes the dimension of $C_{R_1-r}$. Applying Proposition \ref{prop:valuation}, we see  that $ f_i(P) = 0 $ and $ f_i(V) = 0  $ for $ 2 \leqslant i \leqslant k $.
	Therefore the check matrix of $C_{r}$ is given by
	\[
	H=\left[\begin{array}{cccc}
	&  & P & V\\
	\hline * & * & 1 & 1\\
	* & * & 0 & 0\\
	\vdots & \vdots & \vdots & \vdots\\
	* & * & 0 & 0
	\end{array}\right].
	\]
	Let 
	$c:=(0,0,\ldots,0,1,-1)$. Then $Hc^{T}=\mathbf{0}$. Hence, $c\in C_{r}$, and
	$d(C_{r})\leqslant w(c)=2$. Note that $d(C_{r})\geqslant2$, then
	$d(C_{r})=2$.\end{proof}
\begin{example}
	Let us consider the case $ q = 2 $, then $n_{1}=30$, $g=6$, and $R_1=20$. By Proposition
	\ref{prop:Suppose_Bound_d} and Corollary (\ref{cor:basis_rQ}), we obtain the following table.
		
	\begin{table}[h]
		\protect\caption{}
		\renewcommand{\arraystretch}{1.3}
		\begin{tabular}{|c|c|c|c|}
			\hline 
			$ r $ & $\dim$ & basis & $ d $\tabularnewline
			\hline 
			\hline 
			0  & 1 & $1$ & 30\tabularnewline
			\hline 
			2 & 2 & $x$ & 26\tabularnewline
			\hline 
			3 & 3 & $xy$ & 24\tabularnewline
			\hline 
			4 & 4 & $x^{2}$ & 22\tabularnewline
			\hline 
			5 & 5 & $x^{2}y$ & 20\tabularnewline
			\hline 
			6 & 7 & $x^{3},x^{2}y^{2}$ & 18\tabularnewline
			\hline 
			7 & 9 & $x^{3}y$,$x^{2}y^{3}$ & 16\tabularnewline
			\hline 
			8 & 11 & $x^{4}$,$x^{3}y^{2}$ & 14\tabularnewline
			\hline 
			9 & 13 & $x^{4}y$,$x^{3}y^{3}$ & 12\tabularnewline
			\hline 
			10 & 15 & $x^{5}$,$x^{4}y^{2}$ & 10\tabularnewline
			\hline 
		\end{tabular}	
		\begin{tabular}{|c|c|c|c|}
			\hline 
			$ r $ & $\dim$ & basis & $ d $\tabularnewline
			\hline 
			\hline 
			11 & 17 & $x^{5}y$,$x^{4}y^{3}$ & 8\tabularnewline
			\hline 
			12 & 19 & $x^{6}$,$x^{5}y^{2}$ & 6\tabularnewline
			\hline 
			13 & 21 & $x^{6}y$,$x^{5}y^{3}$ & 5\tabularnewline
			\hline 
			14 & 23 & $x^{7}$,$x^{6}y^{2}$ & 2\tabularnewline
			\hline 
			15 & 25 & $x^{7}y$,$x^{6}y^{3}$ & 2\tabularnewline
			\hline 
			16 & 26 & $x^{7}y^{2}$ & 2\tabularnewline
			\hline 
			17 & 27 & $x^{7}y^{3}$ & 2\tabularnewline
			\hline 
			18 & 28 & $x^{8}y^{2}$ & 2\tabularnewline
			\hline 
			19 & 29 & $x^{8}y^{3}$ & 2\tabularnewline
			\hline 
			21 & 30 & $x^{9}y^{3}$ & 1\tabularnewline
			\hline 
		\end{tabular}

			(Note: $C_{1}=C_{0}$, $C_{20}=C_{19}$)
	\end{table}

	The distance with $r=11$ or $13$ is not easy to calculate. By direct computation, $d(C_{13}^{(1)})=5$.	The function 
	\begin{align*}
		\varphi(x) & :=  1+x+xy+x^{2}+x^{2}y+x^{2}y^{2}+x^{3}+0+x^{3}y+0\\
		&   +x^{4}+x^{3}y^{3}+0+x^{4}y^{2}+x^{5}+x^{4}y^{3}+x^{5}y,
	\end{align*}
	achieve the Goppa bound with $r=11$. Then, $d(C_{11}^{(1)})=8$.
\end{example}

\section{the Code {$C({\mathcal{X}},D,rQ+sP+tV)$}}
In this section we study the code
\[C_{r,s,t }:=C_{\mathscr{L}}(D,rQ+sP+tV)  .\]
The length of $  C_{r,s,t } $ is $n_{2}:=\deg(D)=(q^3-1)q^2$.
\begin{prop}[\cite{Stichtenoth}]\label{prop:equivalent}
Suppose $G_{1}$ and $G_{2}$ are divisors with $G_{1}\sim G_{2}$
and $\supp G_{1}\cap \supp D=\supp G_{2}\cap \supp D=\emptyset$, then $C({\mathcal{X}},D,G_{1})$
and $C({\mathcal{X}},D,G_{2})$ are equivalent. 
\end{prop}

Let $C_{r,s}:=C_{r,s,0}$ be the code $ C_{r,s,t} $ with $ t=0 $. We observe that the code $ C_{r,s,t} $ is equivalent to $ C_{r',s'} $ for some  $ r',s' \in \mathbb{Z}$.

\begin{prop}\label{prop:equivalent_rst}
	\begin{enumerate}
		\item  The code $C_{r,s,t}$ is equivalent to $C_{r+t(q+1),s-t(q^{2}+1)}$.
		\item The code $C_{r,s}$ is equivalent to $ C_{r-(q^{2}+q+1),s+(q^{2}+q+1)q}$. 
	\end{enumerate}
	Therefore, $C_{r,s,t}$ can be written as $C_{r^{\prime},s^{\prime}}$  up to equivalence,
	where $0\leqslant r'<(q^{2}+q+1)$, $s' \geqslant 0$.
	\end{prop}
\begin{proof}
	\begin{enumerate}
		\item 	Applying Proposition \ref{prop:valuation}, we obtain
		\[ \Div( xy)=V+(q^{2}+1)P-(q+1)Q ,\] 
		which means that the divisor $ V $ is equivalent to $  (q+1)Q-(q^{2}+1)P $. 
		Hence,
		\[ r Q+s P +t V\sim \left( r+t(q+1)\right) Q+\left( s-t(q^{2}+1) \right) P .\]
		Proposition \ref{prop:equivalent} yields 
		that the code $C_{r,s,t}$ is equivalent to $C_{r+t(q+1),s-t(q^{2}+1)}$.
		\item By Proposition \ref{prop:valuation} again, we get
		\[ \Div( x^{q}y^{q+1})=-(q^{2}+q+1)Q+(q^{2}+q+1)qP.  \]
		So we have
		\[ (q^{2}+q+1)Q\sim (q^{2}+q+1)qP.  \]
		Then,
		\[ rQ+sP\sim\left(  r- (q^{2}+q+1)\right) Q+\left( s+(q^{2}+q+1)q\right) P.  \]
		Hence, the code $C_{r,s}$ is equivalent to $ C_{r-(q^{2}+q+1),s+(q^{2}+q+1)q}$ by Proposition \ref{prop:equivalent}.	
	\end{enumerate}	
	\end{proof}
\begin{prop}\label{prop:dual_rs}
	    Up to equivalence, the dual space of $C_{r,s}$ is 
		\[
		\begin{cases}
		C_{q^{2}-1-r,q^{5}+q^{4}-q^{3}-q^{2}-2q-s} & \text{for \ensuremath{0\leqslant r\leqslant q^{2}-1}},\\
		C_{2q^{2}+q-r,q^{5}+q^{4}-2q^{3}-2q^{2}-3q-s} & \text{for \ensuremath{q^{2}\leqslant r\leqslant q^{2}+q}}.
		\end{cases}
		\]
		Hence,  the  code $C_{q^{2}+q/2,q^{5}/2+q^{4}/2-q^{3}-q^{2}-3q/2}$
		is self-equivalent for even $ q $; and the code 
		$C_{(q^{2}-1)/2,(q^{5}+q^{4}-q^{3}-q^{2}-2q)/2}$ is self-equivalent for odd $q$.		
	 \end{prop}
\begin{proof}
	We follow the notations in the proof of Proposition \ref{prop:The-dual-of},
	\[\Div( \eta )=(q^{4}+q^{3}-q-2)Q-V-P-D  .\]	
	Let $\overline{\eta}:=xy(x^{q}y^{q+1})^{q^{2}-2}\eta $ be a Weil differential on $ \mathcal{X} $. Recall that
	\begin{align*}
	\Div\left((x^{q}y^{q+1})^{q^{2}-2}\right)&=\left(q^{2}-2\right)(q^{3}+q^{2}+q)P\\
	& \quad -\left(q^{2}-2\right)(q^{2}+q+1)Q,
	\end{align*}
	and 
	\[ \Div( xy)=V+(q^{2}+1)P-(q+1)Q .\]	
	Thus, the divisor of $\overline \eta$	is given by	
	\begin{align*}
		\Div( \overline{\eta}) & =  (q^{4}+q^{3}-q-2)Q-V-P-D\\
		& \quad +V+(q^{2}+1)P-(q+1)Q\\
		&  \quad +\left(q^{2}-2\right)\left((q^{3}+q^{2}+q)P-(q^{2}+q+1)Q\right)\\
		& =  -D+(q^{2}-1)Q+(q^{5}+q^{4}-q^{3}-q^{2}-2q)P.
	\end{align*}
    Proposition \ref{prop:equivalent_rst} yields
	\begin{align*}
		 D&-s P-r Q +\Div( \overline{\eta})\\
		 & =  (q^{2}-1-r)Q+(q^{5}+q^{4}-q^{3}-q^{2}-2 q- s )P\\
		& \sim  (2q^{2}+q-r)Q+(q^{5}+q^{4}-2q^{3}-2q^{2}-3 q-s)P.
	\end{align*}
	By Lemma \ref{lem:dual}, we obtain
	\begin{align*}
		C_{r,s}^{\bot} & \cong  C_{q^{2}-1-r,q^{5}+q^{4}-q^{3}-q^{2}-2 q-s}\\
		& \cong  C_{2q^{2}+q-r,q^{5}+q^{4}-2q^{3}-2q^{2}-3 q-s}.
	\end{align*}
\end{proof}
	
	As in Section 2, we set $R_2:=n_{2}+2g-2=q^{5}+q^{4}-q^{2}-3q$. We will be interested in the case, when $0\leqslant rq+s\leqslant R_2$.  Set $  \Omega _{r,s}: = \Omega _{r,s,0}  $, where $ \Omega _{r,s,t} $ is defined in Proposition \ref{prop:basis_rst}. Next we investigate the dimension and the distance of the code $ C_{r,s}  $.
	\begin{prop}
		{{Suppose that}} {$0\leqslant rq+s\leqslant R_2$
			. Then the following holds:}
	\begin{enumerate}
		\item The dimension of $C_{r,s}$ is given by 
		\[
		\dim C_{r,s}=\begin{cases}
		\#\Omega_{r,s} \quad \text{ for \ensuremath{0\leqslant rq+s<n_{2}}}, \\
		n_{2}-\#\Omega_{r,s}^{\bot}=(q^{5}-q^{2})-\#\Omega_{r,s}^{\bot} \\  \qquad \qquad \text{for }n_{2}\leqslant rq+s\leqslant R_2,
		\end{cases}
		\]
	
		where $ \Omega_{r,s}^{\bot} $ is defined by
		 \[ \Omega_{r,s}^{\bot}:=\begin{cases}
		\Omega_{q^{2}-1-r,q^{5}+q^{4}-q^{3}-q^{2}-2q-s} \\
		         \qquad \qquad\text{for \ensuremath{0\leqslant r\leqslant q^{2}-1}},\\
		\Omega_{2q^{2}+q-r,q^{5}+q^{4}-2q^{3}-2q^{2}-3q-s} \\ 
			      \qquad \qquad\text{for \ensuremath{q^{2}\leqslant r\leqslant q^{2}+q}}.
		\end{cases} \]
		
		\item For $q^{4}-3q<rq+s<q^{5}-q^{2}$ we have $\dim C_{r}=qr+s-(q^{4}-3q)/2$.
		\item The minimum distance $d$ of $C_{r}$ satisfies $d\geqslant n_{2}-rq-s=q^{5}-q^{2}-rq-s$.
		\end{enumerate}
	\end{prop}
	\begin{proof}
		\begin{enumerate}
		\item  For $0\leqslant rq+s< n_2$, Proposition \ref{prop:basis_rst} and Equation (\ref{eq:dim}) give 
		\[ \dim C_{r,s} = \dim \mathscr{L}(rQ+sP) = \#\Omega_{r,s} . \] 
		By Proposition \ref{prop:dual_rs},  the dual of the code $ C_{r,s}  $ is equivalent to $ C_{r',s'} $ for some $ r',s' \in \mathbb{Z} $. For $n_{2}\leqslant rq+s\leqslant R_2$, we obtain 
		\begin{align*}
		\dim C_{r,s}&=n_2 - \dim C_{r',s'}\\
		            &=(q^{5}-q^{2})-\#\Omega_{r,s}^{\bot}.
		\end{align*}
		\item Assume that $q^{4}-3q=2g-2<rq+s<n_{2}=q^{5}-q^{2}$. Then, $ \deg(r Q+ s P)>2 g-2 $, so the Riemann-Roch Theorem gives
		\begin{align*}
		\dim C_{r,s}&= \deg ( rQ+sP )+1-g\\
					&=qr+s+1-g\\
		            &=qr+s-(q^{4}-3q)/2.
		\end{align*}
		\item The inequality follows immediately from Goppa bound.
		\end{enumerate}
	\end{proof}
	
	\begin{prop}
		{Suppose $0\leqslant qr+s<n_{2}$.}
		\begin{enumerate}
			\item {If $s=(q^{3}+q^{2}+q)(q\tau-r)$, $0\leqslant(q^{2}+q+1)\tau-qr-r\leqslant q^{3}-1$,
				$0\leqslant q\tau-r$, then $d=n_{2}-rq-s$.}
			
			\item {If $r=0$, $s=q^{2}$, then $d=n_{2}-rq-s$.}
			
			\item {If $s=(q\tau-r)(q^{3}+q^{2}+q)+\lambda q^{2}$, $0\leqslant(q^{2}+q+1)\tau-qr-r\leqslant q^{3}-1-\lambda q^{2}$, $0\leqslant\lambda$, then $d=n_{2}-rq-s$.}
			
			\item {If $r=0$, $s=q^{2}(q^{3}-1-\lambda)$, $\lambda q^{2}\leqslant q^{3}-1$,
				then $d=n_{2}-rq-s$.}
			
			\item {If $s=(q^{3}+q^{2}+q)\left((q-1)q-(q\tau+r)\right)$, $0\leqslant(q^{2}+q+1)\tau+qr+r\leqslant q^{3}-1$,
				then $d=n_{2}-rq-s$.}
		\end{enumerate}
	\end{prop}
	\begin{proof}
		
		\begin{enumerate}
			\item Choose $\kappa:=(q^{2}+q+1)\tau-qr-r$ distinct elements $\alpha_{1},\ldots,\alpha_{\kappa}\in\mathbb{F}_{q^{3}}^{*}$,
			and consider 
			\[
			z:=\left(x^{q}y^{q+1}\right)^{r-q\tau}\prod_{i=1}^{\kappa}(x-\alpha_{i}).
			\]
			The divisor of $z$ is
			\begin{align*}
				\Div( z) 
				& =  (r-q\tau)\left((q^{3}+q^{2}+q)P\right.\\
				& \quad \left. -(q^{2}+q+1)Q\right)\\
				& \quad -q\kappa Q+\sum_{i=1}^{\kappa}D_{\alpha_{i}}\\
				& =  \left(-q\kappa-(r-q\tau)(q^{2}+q+1)\right)Q\\
				& \quad-(q\tau-r)(q^{3}+q^{2}+q)P+\sum_{i=1}^{\kappa}D_{\alpha_{i}}\\
				& =  -r Q -(q\tau-r)(q^{3}+q^{2}+q)P
				 +\sum_{i=1}^{\kappa}D_{\alpha_{i}}.
			\end{align*}
			
			\item Fix an element $ \beta \in \mathbb{F}_{q^3}^* $. We consider 
			\[
			z:=y^{-1}(y-\beta).
			\]
			Recall that the divisor of $y-\beta$ is 
			\[ \Div( y-\beta)=-Q+D_{\beta}-qV. \]
			 Then 
			\begin{align*}
				\Div( z) & =  -Q+D_{\beta}-qV+q V-q ^{2}P+Q\\
				& =  D_{\beta}-q^{2}P.
			\end{align*}
			
			\item Let $ \kappa : = (q^{2}+q+1)\tau-qr-r $ with $0\leqslant \kappa \leqslant q^{3}-1-\lambda q^{2}$. Choose $ \lambda  $ distinct elements $ \beta _1, \ldots,\beta_\lambda $ in $ \mathbb{F}_{q^3} ^*$, and $\kappa$ distinct elements $\alpha_1, \ldots, \alpha_\kappa $ in the set  $A:=\{\alpha\in \mathbb{F}_{q^3} ^* |(\alpha,\beta) \not \in D_{\beta_i} , i=1,\ldots, \lambda \}$. 
			Consider 
			\[
			z:=y^{-\lambda}\prod_{i=1}^{\lambda}(y-\beta_{i})(x^{q+1}y^{q})^{r-q\tau}\prod_{i=1}^{\kappa}(x-\alpha_{i}).
			\]
			Then the divisor of $z$ is
			\begin{align*}
				  \Div( z) &=  -\lambda q^{2}P+\sum_{i=1}^{\lambda} D_{\beta_{i}}-q\kappa Q+\sum_{i=1}^{\kappa}D_{\alpha_{i}}\\
				&  \quad +(r-q\tau)\left((q^{3}+q^{2}+q)P-(q^{2}+q+1)Q\right)\\
				& =  -rQ-\left((q\tau-r)(q^{3}+q^{2}+q)+\lambda q^{2}\right)P\\
				& \quad +\sum_{i=1}^{\kappa}D_{\alpha_{i}} +\sum_{i=1}^{\lambda}D_{\beta_{i}}.
			\end{align*}
		
			\item Suppose that $ \mathbb{F}_{q^3}^*=\{\alpha_1,\ldots, \alpha_{q^3-1}\} $, and $ \beta_1, \ldots, \beta_ {\lambda}$ are distinct elements in $ \mathbb{F}_{q^3}^* $. Consider 
			\[
			z:=y^{\lambda}\prod_{i=1}^{\lambda}(y-\beta_{i})^{-1}(x^{q}y^{q+1})^{q-q^{2}}\prod_{i=1}^{q^{3}-1}(x-\alpha_{i}),
			\]
			where $\lambda q^{2}\leqslant q^{3}-1$. Then its divisor is 
			\begin{align*}
			 	\Div( z)
			 & =  -\sum_{i=1}^{\lambda } D_{\beta_{i}}+\lambda q^{2}P+q(q^{3}-1)Q\\
				& \quad -q^{2}(q^{3}-1)P -q(q^{3}-1)Q+D\\
				& =  D-\sum_{i=1}^{\lambda } D_{\beta_{i}}-q^{2}(q^{3}-1-\lambda)P.
			\end{align*}
			
			\item Let $\epsilon=(q^{2}+q+1)\tau+qr+r$. Suppose that $ \mathbb{F}_{q^3}^*=\{\beta_1,\ldots, \beta_{q^3-1}\} $, and $ \alpha_1, \ldots, \alpha_ {\epsilon}$ are distinct elements in $ \mathbb{F}_{q^3}^* $. Consider
			\[
			z:=y^{q^{3}-1}\prod_{i=1}^{q^{3}-1}(y-\beta_{i})(x^{q}y^{q+1})^{q\tau+r}\prod_{i=1}^{\epsilon}(x-\alpha_{i})^{-1}.
			\]
			Its divisor is 
			\begin{align*}				  
	\Div( z)	&  =  \sum_{i=1}^{q^3-1 } D_{\beta_{i}}-(q^{3}-1)q^{2}P\\
	            & \quad -(q^{2}+q+1)(q\tau+r)Q\\
				& \quad +(q^{2}+q+1)q(q\tau+r)P    +q\epsilon Q-\sum_{i=1}^{\epsilon} D_{\alpha_i}\\
				& =  D-\sum_{i=1}^{\epsilon} D_{\alpha_i}-rQ\\
				& \quad -(q^{2}+q+1)q\left((q-1)q-(q\tau+r)\right)P.
			\end{align*}
			
		\end{enumerate}
	\end{proof}
	 For a fixed $r$, we define the Weierstrass set  
	 \[
	  H_{r}=\{s\in \mathbb {Z}|\mathscr {L}(rQ+sP)\not=\mathscr {L}(rQ+(s-1)P)\}. 
	  \] 
	 Proposition \ref{prop: basis} yields
	 \[
	  H_{r}=\{s\in \mathbb {Z}|\Omega_{r,s}\not=\Omega_{r,s-1}\}.
	  \] 
	  	 Assume that $ f \in \mathscr {L}(rQ+sP)\backslash \mathscr {L}(rQ+(s-1)P) $,  and $ g \in\mathscr {L}(s'P)\backslash \mathscr {L}((s'-1)P)   $. Then, 
	  	  \[ f g \in \mathscr {L}\left(rQ+(s+s')P\right)\backslash \mathscr {L}\left(rQ+(s+s'-1)P\right). \]
	 This implies 
	  \[ H_0 +H_r \subseteq H_r.\]
	 Define the set 
	\[ H_{r}^{*}=\{s\in \mathbb {Z}|C_{r,s}\not=C_{r,s-1}\} .\] 
	So we can restrict to consider the codes $  C_{r,s}  $ with $ s \in  H_{r}^{*} $. It is easy to see that $H^*_r$ consists of $ n_2 $ elements. Let us write $ H^*_r = \{ s_1^{*}< s_2^{*}< \ldots < s_{n_2}^{*} \} $. Then $ \dim(C_{r,s_i^{*}})=i $.
	Clearly, $H_{r}^{*}\subseteq H_{r}$ and $H_{r}^{*}\cap{\left\{ s|rq+s<n_{2}\right\} }=H_{r}\cap{\left\{ s|rq+s<n_{2}\right\} }$.	
	For $  s \in \mathbb {Z} $ satisfying $rq+s\geqslant n_{2}$, then s $\in H_{r}^{*}$ if and only if $s\in H_{r}^{\bot}$, 
	where $H_{r}^{\bot}$ is defined by  
	\begin{equation}\label{eq:H_r_bot}
	 H_{r}^{\bot}:=\begin{cases}
	q^{5}+q^{4}-q^{3}-q^{2}-2q+1-H_{q^{2}-1-r} \\
	\qquad \qquad  \text{for \ensuremath{0\leqslant r\leqslant q^{2}-1}},\\
	q^{5}+q^{4}-2q^{3}-2q^{2}-3q+1-H_{2q^{2}+q-r} \\
	\qquad \qquad \text{for \ensuremath{q^{2}\leqslant r\leqslant q^{2}+q}}.
	\end{cases}
	\end{equation} 
	We remark that besides the Goppa bound there are several bounds available to estimate the minimum distance of a code. One of the most interesting is the order bound. We can follow the version of \cite{Beelen}, which is briefly explained below.
	For $i=1 , \ldots, n_2$, let 
	\[ \Lambda_i^r: =\{(a,b)|a \in H_0,b \in H_r,a+b=s_i^* \in H_r^{*}\} .\]
	By using the notation in \cite{Beelen}, we consider the infinite sequence   	$S=P, P, P,\ldots$, then the minimum distance of the dual code of $ C_{r, s} $ verifies
\[ 	d(C_{r,s}^{\bot} )\geqslant  d_S(rQ+sP):=\min_{s_i^*> s} \{\#\Lambda_i^{r} \}. \]

	\begin{example}
		Firstly, we consider the  code $C_{5,s}$ over $ \mathbb{F}_8 $ of length  $n_{2}=28$. The genus of curve $\mathcal X$ is  $g=6$. We find that the dual code of $C_{5,s}$ is equivalent to $ C_{5,18-s}$ by Proposition \ref{prop:equivalent_rst}. According to  Proposition \ref{prop:basis_rst} and Equation (\ref{eq:H_r_bot}), we get  
	\[ 	H_5 =\{ -6, -5, -2, -1, 0, 1 ,2 , 3 ,\ldots \}, \]
		and
		\begin{align*}
		H_5^{*} =&\{ -6, -5, -2, -1, 0, 1 ,2 , 3 ,\\
		     &  \ldots, 16, 17, 18, 19, 20,21 ,24 , 25  \}.
		\end{align*}
	By direct computation, we obtain the following table.
	\begin{table}[H]
		\protect\caption{}
		\label{t:table}
		\renewcommand{\arraystretch}{1.3}
		\begin{tabular}{|c|c|c|c|}
			\hline 
			$s$ & $\dim$ & basis & $d$\tabularnewline
			\hline 
			\hline 
			-6 & 1 & $x^{2}y$ & 28\tabularnewline
			\hline 
			-5 & 2 & $xy$ & 24\tabularnewline
			\hline 
			-2 & 3 & $x^{2}$ & 24\tabularnewline
			\hline 
			-1 & 4 & $x$ & 20\tabularnewline
			\hline 
			0 & 5 & $1$ & 18\tabularnewline
			\hline 
			1 & 6 & $x^{3}y^{-1}$ & 18\tabularnewline
			\hline 
			2 & 7 & $x^{2}y^{-1}$ & 16\tabularnewline
			\hline 
			3 & 8 & $xy^{-1}$ & 16\tabularnewline
			\hline 
			4 & 9 & $y^{-1}$ & 15\tabularnewline
			\hline 
			5 & 10 & $x^{3}y^{-2}$ & 13\tabularnewline
			\hline 
			6 & 11 & $x^{2}y^{-2}$ & 12\tabularnewline
			\hline 
			7 & 12 & $xy^{-2}$ & 12\tabularnewline
			\hline 
			8 & 13 & $y^{-2}$ & 11\tabularnewline
			\hline 
			9 & 14 & $x^{-1}y^{-2}$ & 10\tabularnewline
			\hline 
		\end{tabular}
		\begin{tabular}{|c|c|c|c|}
			\hline 
			$s$ & $\dim$ & basis & $d$\tabularnewline
			\hline 
			\hline 
			10 & 15 & $x^{2}y^{-3}$ & 8\tabularnewline
			\hline 
			11 & 16 & $xy^{-3}$ & 8\tabularnewline
			\hline 
			12 & 17 & $y^{-3}$ & 8\tabularnewline
			\hline 
			13 & 18 & $x^{-1}y^{-3}$ & 7\tabularnewline
			\hline 
			14 & 19 & $x^{-2}y^{-3}$ & 4\tabularnewline
			\hline 
			15 & 20 & $xy^{-4}$ & 4\tabularnewline
			\hline 
			16 & 21 & $y^{-4}$ & 4\tabularnewline
			\hline 
			17 & 22 & $x^{-1}y^{-4}$ & 4\tabularnewline
			\hline 
			18 & 23 & $x^{-2}y^{-4}$ & 3\tabularnewline
			\hline 
			19 & 24 & $xy^{-5}$ & 3\tabularnewline
			\hline 
			20 & 25 & $y^{-5}$ & 3\tabularnewline
			\hline 
			21 & 26 & $x^{-1}y^{-5}$ & 2\tabularnewline
			\hline 
			24 & 27 & $y^{-6}$ & 2\tabularnewline
			\hline 
			25 & 28 & $x^{-1}y^{-6}$ & 1\tabularnewline
			\hline 
		\end{tabular}
	\end{table}		
		Comparing Table \ref{t:table} with the reference \cite{Grassl}, we find the following codes
		over $\mathbb{F}_{8}$ with the best known parameters:
		
		$ {[}28,1,28{]} $, ${[}28,2,24{]} $, $ {[}28,3,24{]} $, $ {[}28,8,16{]} $, $ {[}28,12,12{]} $,  ${[}28,17,8{]} $, $ {[}28,25,3{]} $, $ {[}28,26,2{]} $, $ {[}28,27,2{]} $, $ {[}28,28,1{]} $. 
		
		Table \ref{t:table} enables one to construct the codes explicitly. For instance, a $ {[}28,8,16{]} $-code is constructed by the basis $x^{2}y$, $xy$, $x^{2}$, $x$, $1$, $x^{3}y^{-1}$, $x^{2}y^{-1}$, $xy^{-1}$.
	\end{example}
	
	\begin{example}
		Fix $r=0$, we consider the one-point code $C_{0,s}$ over $\mathbb{F}_{8}$ which is dual to $C_{3,32-s}$ up to equivalence. Counting the lattice points in both sets $ \Omega_{0,s} $ and $ \Omega_{3,s} $, we find the Weierstrass sets
		\[  H_{0} = \{0 ,4, 7, 8, 9, 11, 12, 13, 14,  \ldots \} ,\]
		and
		\[  H_{3} = \{-5 ,-1, 0, 2, 3, 4, 6, 7, 8,  \ldots \} . \]
		Using Equation (\ref{eq:H_r_bot}), the set $  H_{0}^{*}  $ is
		\begin{align*}
			H_{0}^{*} =& \{ 0 ,4, 7, 8, 9, 11, 12, 13, 14, \\
			 & \ldots, 26,27, 29,  30, 31, 33, 34, 38\} .
		\end{align*} 
		\begin{table}[H]
			\protect\caption{}	
			\renewcommand{\arraystretch}{1.3}		
			\begin{tabular}{|c|c|c|c|}
				\hline 
				$s$ & $\dim$ & basis & $d$\tabularnewline
				\hline 
				\hline 
				0 & 1 & 1 & 28\tabularnewline
				\hline 
				4 & 2 & $y^{-1}$ & 24\tabularnewline
				\hline 
				7 & 3 & $xy^{-2}$ & 21\tabularnewline
				\hline 
				8 & 4 & $y^{-2}$ & 20\tabularnewline
				\hline 
				9 & 5 & $x^{-1}y^{-2}$ & 19\tabularnewline
				\hline 
				11 & 6 & $xy^{-3}$ & 18\tabularnewline
				\hline 
				12 & 7 & $y^{-3}$ & 16\tabularnewline
				\hline 
				13 & 8 & $x^{-1}y^{-3}$ & 15\tabularnewline
				\hline 
				14 & 9 & $x^{-2}y^{-3}$ & 14\tabularnewline
				\hline 
				15 & 10 & $xy^{-4}$ & 13\tabularnewline
				\hline 
				16 & 11 & $y^{-4}$ & 12\tabularnewline
				\hline 
				17 & 12 & $x^{-1}y^{-4}$ & 12\tabularnewline
				\hline 
				18 & 13 & $x^{-2}y^{-4}$ & 11\tabularnewline
				\hline 
				19 & 14 & $xy^{-5}$ & 9\tabularnewline
				\hline 
			\end{tabular}
			\begin{tabular}{|c|c|c|c|}
				\hline 
				$s$ & $\dim$ & basis & $d$\tabularnewline
				\hline 		
				\hline 
				20 & 15 & $y^{-5}$ & 8\tabularnewline
				\hline 
				21 & 16 & $x^{-1}y^{-5}$ & 7\tabularnewline
				\hline 
				22 & 17 & $x^{-2}y^{-5}$ & 7\tabularnewline
				\hline 
				23 & 18 & $x^{-3}y^{-5}$ & 6\tabularnewline
				\hline 
				24 & 19 & $y^{-6}$ & 4\tabularnewline
				\hline 
				25 & 20 & $x^{-1}y^{-6}$ & 4\tabularnewline
				\hline 
				26 & 21 & $x^{-2}y^{-6}$ & 4\tabularnewline
				\hline 
				27 & 22 & $x^{-3}y^{-6}$ & 4\tabularnewline
				\hline 
				29 & 23 & $x^{-1}y^{-7}$ & 4\tabularnewline
				\hline 
				30 & 24 & $x^{-2}y^{-7}$ & 3\tabularnewline
				\hline 
				31 & 25 & $x^{-3}y^{-7}$ & 3\tabularnewline
				\hline 
				33 & 26 & $x^{-1}y^{-8}$ & 2\tabularnewline
				\hline 
				34 & 27 & $x^{-2}y^{-8}$ & 2\tabularnewline
				\hline 
				38 & 28 & $x^{-2}y^{-9}$ & 1\tabularnewline
				\hline 
			\end{tabular}
		\end{table}	
		 We find a $ {[}28,12,12{]} $-code
		 over $\mathbb{F}_{8}$ with the best known parameters  as in the previous example.	
	
	\end{example}
	
	\begin{example}
		Let us consider the code $C_{4,s}$ over $\mathbb{F}_{27}$, with $g=37$, $n_{2}=234$. The code $C_{4,s}$
		is dual to $C_{4,282-s}$ up to equivalence. The Weierstrass set for $r=4$ is  $H_4=\{-10$, $-1$, $0$, $ 8$, $ 9$, $16$, $17$, $18$, $19$, $25$, $26$, $27$, $28$, $29$, $34$, $35$, $36$, $37$, $38$, $39$, $42$, $43$, $44$, $45$, $46$, $47$, $48$, $51$, $52$, $53$, $54$, $55$, $56$, $57$, $58$, $60$, $61$, $62$, $63$, $64$, $65$, $66$, $67$, $68$, $69$, $70$, $71$, $72$, $\ldots \}$. 
		
		According to Equation (\ref{eq:H_r_bot}), we obtain
	$H_4=\{-10$, $-1$, $0$, $ 8$, $ 9$, $16$, $17$, $18$, $19$, $25$, $26$, $27$, $28$, $29$, $34$, $35$, $36$, $37$, $38$, $39$, $42$, $43$, $44$, $45$, $46$, $47$, $48$, $51$, $52$, $53$, $54$, $55$, $56$, $57$, $58$, $60$, $61$, $62$, $63$, $64$, $65$, $66$, $67$, $68$, $69$, $70$, $71$, $72$, $\ldots$ , $214$, $215$, $216$, $217$, $218$, $219$, $220$, $221$, $222$, $223$, $225$, $226$, $227$, $228$, $229$, $230$, $231$, $232$, $235$, $236$, $237$, $238$, $239$, $240$, $241$, $244$, $245$, $246$, $247$, $248$, $249$, $254$, $255$, $256$, $257$, $258$, $264$, $265$, $266$, $267$, $274$, $275$, $283$, $284$, $293 \}$. Computing the order bound, we get $  d(C_{4,165})\geqslant 59 $. So we obtain  a new record-giving $[234,141,\geqslant 59]$-code
 over $\mathbb{F}_{27}$ (according
	to the tables \cite{MinT}). 
	Similarly, for general $r$, we can find  the record-giving codes  over $\mathbb{F}_{27}$  as follows: $[234,143,\geqslant 57]$, $[234,144,\geqslant 56]$, $[234,145,\geqslant 55]$.

		\end{example}


%





\ifCLASSOPTIONcaptionsoff
  \newpage
\fi



%
\bibliographystyle{IEEEtran}
\bibliography{paper}

\begin{thebibliography}{10}
\providecommand{\url}[1]{#1}
\csname url@samestyle\endcsname
\providecommand{\newblock}{\relax}
\providecommand{\bibinfo}[2]{#2}
\providecommand{\BIBentrySTDinterwordspacing}{\spaceskip=0pt\relax}
\providecommand{\BIBentryALTinterwordstretchfactor}{4}
\providecommand{\BIBentryALTinterwordspacing}{\spaceskip=\fontdimen2\font plus
\BIBentryALTinterwordstretchfactor\fontdimen3\font minus
  \fontdimen4\font\relax}
\providecommand{\BIBforeignlanguage}[2]{{%
\expandafter\ifx\csname l@#1\endcsname\relax
\typeout{** WARNING: IEEEtran.bst: No hyphenation pattern has been}%
\typeout{** loaded for the language `#1'. Using the pattern for}%
\typeout{** the default language instead.}%
\else
\language=\csname l@#1\endcsname
\fi
#2}}
\providecommand{\BIBdecl}{\relax}
\BIBdecl

\bibitem{Stichtenoth}
H.~Stichtenoth, \emph{Algebraic function fields and codes}.\hskip 1em plus
  0.5em minus 0.4em\relax Springer Science \& Business Media, 2009, vol. 254.

\bibitem{Tiersma}
H.~Tiersma, ``Remarks on codes from {Hermitian} curves (corresp.),''
  \emph{Information Theory, IEEE Transactions on}, vol.~33, no.~4, pp.
  605--609, Jul 1987.

\bibitem{Guruswami}
V.~Guruswami and M.~Sudan, ``Improved decoding of {Reed-Solomon} and
  algebraic-geometric codes,'' in \emph{Foundations of Computer Science, 1998.
  Proceedings. 39th Annual Symposium on}, Nov 1998, pp. 28--37.

\bibitem{Yang}
K.~Yang and P.~V. Kumar, ``On the true minimum distance of {Hermitian} codes,''
  in \emph{Coding theory and algebraic geometry}.\hskip 1em plus 0.5em minus
  0.4em\relax Springer, 1992, pp. 99--107.

\bibitem{Yang2}
K.~Yang, P.~V. Kumar, and H.~Stichtenoth, ``On the weight hierarchy of
  geometric {Goppa} codes,'' \emph{Information Theory, IEEE Transactions on},
  vol.~40, no.~3, pp. 913--920, 1994.

\bibitem{Homma}
M.~Homma and S.~Kim, ``\BIBforeignlanguage{English}{Toward the determination of
  the minimum distance of two-point codes on a {Hermitian} curve},''
  \emph{\BIBforeignlanguage{English}{Designs, Codes and Cryptography}},
  vol.~37, no.~1, pp. 111--132, 2005.

\bibitem{Homma2}
------, ``\BIBforeignlanguage{English}{The two-point codes on a {Hermitian}
  curve with the designed minimum distance},''
  \emph{\BIBforeignlanguage{English}{Designs, Codes and Cryptography}},
  vol.~38, no.~1, pp. 55--81, 2006.

\bibitem{Homma3}
------, ``\BIBforeignlanguage{English}{The two-point codes with the designed
  distance on a {Hermitian} curve in even characteristic},''
  \emph{\BIBforeignlanguage{English}{Designs, Codes and Cryptography}},
  vol.~39, no.~3, pp. 375--386, 2006.

\bibitem{Homma4}
------, ``\BIBforeignlanguage{English}{The complete determination of the
  minimum distance of two-point codes on a {Hermitian} curve},''
  \emph{\BIBforeignlanguage{English}{Designs, Codes and Cryptography}},
  vol.~40, no.~1, pp. 5--24, 2006.

\bibitem{Park}
S.~Park, ``\BIBforeignlanguage{English}{Minimum distance of {Hermitian}
  two-point codes},'' \emph{\BIBforeignlanguage{English}{Designs, Codes and
  Cryptography}}, vol.~57, no.~2, pp. 195--213, 2010.

\bibitem{Korchmaros}
G.~Korchm\'{a}ros and G.~Nagy, ``{Hermitian} codes from higher degree places,''
  \emph{Journal of Pure and Applied Algebra}, vol. 217, no.~12, pp. 2371 --
  2381, 2013.

\bibitem{Ballico}
E.~Ballico and A.~Ravagnani, ``A zero-dimensional cohomological approach to
  {Hermitian} codes,'' 2013.

\bibitem{Couvreur}
A.~Couvreur, ``The dual minimum distance of arbitrary-dimensional
  algebraic--geometric codes,'' \emph{Journal of Algebra}, vol. 350, no.~1, pp.
  84--107, 2012.

\bibitem{Bulygin}
S.~Bulygin, ``Generalized {Hermitian} codes over ${GF}(2^{r})$,''
  \emph{Information Theory, IEEE Transactions on}, vol.~52, no.~10, pp.
  4664--4669, Oct 2006.

\bibitem{Munuera}
C.~Munuera, A.~Sep\'{u}lveda, and F.~Torres,
  ``\BIBforeignlanguage{English}{Generalized {Hermitian} codes},''
  \emph{\BIBforeignlanguage{English}{Designs, Codes and Cryptography}},
  vol.~69, no.~1, pp. 123--130, 2013.

\bibitem{Bassa}
A.~Bassa, P.~Beelen, A.~Garcia, and H.~Stichtenoth, ``Towers of function fields
  over non-prime finite fields,'' \emph{arXiv preprint arXiv:1202.5922}, 2012.

\bibitem{Gilbert}
W.~J. Gilbert and S.~A. Vanstone, \emph{An introduction to mathematical
  thinking: algebra and number systems}.\hskip 1em plus 0.5em minus 0.4em\relax
  Pearson Prentice Hall, 2005.

\bibitem{Rosen}
K.~H. Rosen, \emph{Elementary number theory and its applications}.\hskip 1em
  plus 0.5em minus 0.4em\relax Reading, Mass., 1993.

\bibitem{Haigh}
G.~Haigh, ``{A 'natural' approach to Pick's theorem},'' \emph{The Mathematical
  Gazette}, pp. 173--177, 1980.

\bibitem{Varberg}
D.~E. Varberg, ``Pick's theorem revisited,'' \emph{American Mathematical
  Monthly}, pp. 584--587, 1985.

\bibitem{Beelen}
P.~Beelen, ``The order bound for general algebraic geometric codes,''
  \emph{Finite Fields and Their Applications}, vol.~13, no.~3, pp. 665 -- 680,
  2007.

\bibitem{Grassl}
\BIBentryALTinterwordspacing
M.~Grassl, ``Bounds on the minimum distance of linear codes and quantum
  codes.'' Accessed on 2015-03-08. [Online]. Available:
  \url{http://www.codetables.de.}
\BIBentrySTDinterwordspacing

\bibitem{MinT}
\BIBentryALTinterwordspacing
MinT, ``Online database for optimal parameters of $ (t,m,s) $-nets, $ (t,s)
  $-sequences, orthogonal arrays, and linear codes.'' Accessed on 2015-03-08.
  [Online]. Available: \url{http://mint.sbg.ac.at.}
\BIBentrySTDinterwordspacing

\end{thebibliography}

%






\end{document}